\documentclass[a4paper,UKenglish,cleveref,nameinlink, autoref, cleveref, thm-restate]{lipics-v2021}

\pdfoutput=1 %
\hideLIPIcs

\title{The Threshold Problem for Hypergeometric Sequences with Quadratic Parameters}

\titlerunning{The Threshold Problem for Hypergeometric Sequences with Quadratic Parameters}

\author{George Kenison}{School of Computer Science and Mathematics, Liverpool John Moores University, Liverpool, UK}{g.j.kenison@ljmu.ac.uk}{}{}

\authorrunning{G.~Kenison} %

\Copyright{George Kenison} %



\nolinenumbers %

\usepackage{amsmath, amsthm, amssymb, mathrsfs, calc}
\usepackage{mathtools}
\usepackage{url}
\usepackage{mleftright}
\usepackage{microtype}
\usepackage[small,basic]{complexity}
\usepackage{aliascnt}
\usepackage{hyperref}

\usepackage[T1]{fontenc}

\newcommand{\appref}[1]{\hyperref[#1]{Appendix~\ref*{#1}}}

\newaliascnt{ass}{theorem}
\newtheorem{assumption}[ass]{Property}

\aliascntresetthe{ass}

\renewcommand{\Re}{\operatorname{Re}}

\renewcommand{\S}{\textsection}

\renewcommand{\epsilon}{\ensuremath\varepsilon}

\renewcommand{\Lambda}{\ensuremath\mu}

\renewcommand{\R}{\mathbb{R}}
\newcommand{\Z}{\mathbb{Z}}
\newcommand{\N}{\mathbb{N}}

\newcommand{\Q}{\mathbb{Q}}

\renewcommand{\Re}{\operatorname{Re}}

\newcommand{\oalpha}{\smash{\overline{\alpha}}\vphantom{\alpha}}

\newcommand*\wc{{}\cdot{}}

\providecommand*{\eu}%
{\ensuremath{\mathrm{e}}}
\providecommand*{\iu}%
{\ensuremath{\mathrm{i}}}

\newcommand{\seq}[3][{}]{\langle #2 \rangle_{#3}^{#1}}

\date{}

\newcounter{contributions}

\begin{document}

\maketitle

\begin{abstract}
Hypergeometric sequences are rational-valued sequences that satisfy first-order linear recurrence relations with polynomial coefficients; that is,
\(\seq[\infty]{u_n}{n=0}\) is hypergeometric if it satisfies a first-order linear recurrence of the form \(p(n)u_{n+1} =  q(n)u_{n}\) with polynomial coefficients \(p,q\in\Z[x]\) and \(u_0\in\Q\).

In this paper, we consider the Threshold Problem for hypergeometric sequences: given a hypergeometric sequence \(\seq[\infty]{u_n}{n=0}\) and a threshold \(t\in\Q\), determine whether \(u_n \ge t\) for each \(n\in\N_0\).
We  establish decidability for the Threshold Problem under the assumption that  the coefficients \(p\) and \(q\) are monic polynomials whose roots lie in an imaginary quadratic extension of \(\Q\).
We also establish conditional decidability results; for example,
under the assumption that the coefficients \(p\) and \(q\) are monic polynomials whose roots lie in any number of quadratic extensions of \(\Q\), the Threshold Problem is decidable subject to the truth of Schanuel's conjecture.
Finally, we show how our approach both recovers and extends some of the recent decidability results on the Membership Problem for hypergeometric sequences with quadratic parameters.
\end{abstract}

\section{Introduction}

\subsection*{Background}
The \emph{Threshold Problem} is a fundamental open decision problem in automated verification that asks to determine whether every term in a recursively defined sequence is bounded from below by a given value (commonly, the \emph{threshold}).
The Threshold Problem appears under many guises across the computational and mathematical sciences with applications in fields as diverse as software verification, probabilistic model checking, combinatorics, and formal languages (we refer the interested reader to the discussion in \cite{ouaknine2014positivity} and the references therein).

The inputs for the Threshold Problem are a recursively defined sequence
\(\seq[\infty]{u_n}{n=0}\subseteq \Q\) and a threshold \(t\in\Q\).
(Hereafter, we shall use tuple notation \((\seq[\infty]{u_n}{n=0},t)\) as shorthand for a given problem instance.) 
Threshold then asks to determine 
 whether  \(u_n \ge t\) for each \(n\in\N_0\).
 Arguably, the variant of the Threshold Problem that has received the most attention in automated verification is the Positivity Problem for \emph{C-finite sequences} (those sequences that obey a linear recurrence relation with constant coefficients).
 Therein, Positivity sets as a threshold \(t=0\) and so asks whether every term in a C-finite sequence is non-negative.

Herein we consider the hypergeometric subclass of \emph{P-finite sequences} (those sequences that satisfy a linear recurrence relation with polynomial coefficients) \cite{kauers2011tetrahedron}.
Recall that a \emph{hypergeometric sequence} is a rational-valued first-order linear recurrence sequence with polynomial coefficients;
that is to say, a sequence \(\seq[\infty]{u_n}{n=0}\subseteq \Q\) that satisfies a relation of the form
	\begin{equation} \label{eq:rec}
		p(n)u_{n+1} = q(n)u_n
	\end{equation}
where \(p,q\in\Z[x]\) and \(p(x)\) has no non-negative integer zeros.
By the latter assumption on \(p(x)\), the recurrence relation \eqref{eq:rec} uniquely defines an infinite sequence of rational numbers once the initial value \(u_0\in\Q\) is specified.
 For a hypergeometric sequence \(\seq[\infty]{u_n}{n=0}\) satisfying \eqref{eq:rec}, we call the roots of the polynomial \(pq\) the sequence's \emph{parameters}.
Hypergeometric sequences and their associated generating functions, the hypergeometric series, are commonplace in fields such as numerical analysis and analytic combinatorics~\cite{FS09,kauers2011tetrahedron}.

In this paper, we consider the Threshold Problem for hypergeometric sequences.
Na{\"\i}vely, we might construe that decidability of the Threshold Problem in this setting is easily settled.
Consider an instance of the Threshold Problem \((\seq[\infty]{u_n}{n=0}, t)\). %
Without loss of generality, we can assume that \(\seq[\infty]{u_n}{n=0}\) either diverges to infinity or converges to a non-zero limit (further explanation behind this assumption is given in the Preliminaries).
Suppose that \(\seq[\infty]{u_n}{n=0}\) converges to a limit not equal to \(t\).
From the form of the recurrence relation in \eqref{eq:rec}, we  can
compute a bound \(B\) such that if \(n>B\) then \(u_n > t\) or \(u_n < t\).
Similar deductions handle the case that \(\seq[\infty]{u_n}{n=0}\) diverges to infinity.
In the case that the limit of  \(\seq[\infty]{u_n}{n=0}\) is the threshold \(t\),
we can compute a similar bound based on the fact that the convergence to \(t\) is eventually monotonic.
It follows that, in each case, the Threshold Problem reduces to exhaustively checking whether \(u_n \ge t\) for each \(n\in\{0,1,\ldots, B\}\).
Unfortunately this reasoning does not suffice to decide the Threshold Problem.
Indeed, we do not know how to decide whether a generic hypergeometric sequence converges to a given rational limit.
Further, such convergence questions are intricately linked to open problems concerning algebraic relations for the gamma function (we give further details below).

\subsection*{Contributions}

Our primary contributions are:
	\begin{alphaenumerate}
		\item The Threshold Problem for hypergeometric sequences whose polynomial coefficients are monic and split over an imaginary quadratic field are decidable (\autoref{thm:mainquadratic}).
		\item The Threshold Problem for hypergeometric sequences whose polynomial coefficients are monic and each irreducible factor of \(pq\) is either linear or quadratic is decidable subject to the truth of Schanuel's conjecture (\autoref{cor:quadraticSC}).
		\setcounter{contributions}{\value{enumi}}
	\end{alphaenumerate}
We delay a formal statement of Schanuel's conjecture to the Preliminaries.
For our conditional decidability results, we note that only termination is conditional on Schanuel's conjecture and that correctness of our procedure is unconditional (\autoref{rem:MacintyreWilkie}).
\autoref{cor:quadraticSC} follows from the more general result:
	\begin{alphaenumerate}
\setcounter{enumi}{\value{contributions}}
	\item The Threshold Problem for hypergeometric sequences whose monic polynomial coefficients possess \autoref{ass:symmetry} is decidable subject to the truth of Schanuel's conjecture (\autoref{thm:mainclass} in \autoref{ssec:4.1}).
\setcounter{contributions}{\value{enumi}}
	\end{alphaenumerate}
Polynomials with \autoref{ass:symmetry} (\autoref{ssec:4.1}) lead to classes of hypergeometric sequences with unnested radical and cyclotomic parameters.

Our secondary contribution concerns the \emph{Membership Problem} for hypergeometric sequences.
Given a hypergeometric sequence \(\seq[\infty]{u_n}{n=0}\) and target \(t\in\Q\), Membership asks to determine whether there is an \(n\in\N_0\) for which \(u_n=t\).
\begin{alphaenumerate}
\setcounter{enumi}{\value{contributions}}
	\item For classes of hypergeometric sequences where we establish (un)conditional decidability of the Threshold Problem, we also obtain (un)conditional decidability of the Membership Problem.
	This contribution is a straightforward corollary of the following observation: for hypergeometric sequences, decidability of the Membership Problem reduces to that of the Threshold Problem (\autoref{prop:reduction}).
\end{alphaenumerate}
We note this secondary contribution both recovers and extends some of the recent results in work by Kenison et al.~\cite{kenison2023membership}.
For the avoidance of doubt, we break new ground for the Membership Problem.
For example, we have conditional decidability of the Membership Problem for hypergeometric sequences with unnested radical and cyclotomic parameters (both classes fall outside of the remit of the previous works). 
A concrete subclass is given by those  hypergeometric sequences whose polynomial coefficients are of the form \((x^2-\ell_1)(x^2-\ell_2)\cdots (x^2-\ell_d)\) where \(\ell_1,\ldots, \ell_d\in\Z\).

\subsection*{Approach}
As previously mentioned, an obstacle that prevents us from settling decidability of the Threshold Problem for hypergeometric sequences is determining whether a given hypergeometric sequence converges to some rational limit.

To each hypergeometric sequence
 \(\seq[\infty]{u_n}{n=0}\) satisfying \eqref{eq:rec}, we associate the \emph{shift quotient} \(r(x):= q(x)/p(x) \in\Q(x)\).
It is clear that the terms of  \(\seq[\infty]{u_n}{n=0}\) are given by a sequence of partial products such that the \(n\)th term is given thus:
		\( u_n = u_0 \cdot \prod_{k=0}^n r(k) \).
Without loss of generality, we can a normalise a sequence with \(u_0\neq 0\) by assuming that \({u_0=1}\).
(We note that the Threshold Problem is trivial to decide when \(u_0=0\).)
Thus our consideration of the Threshold Problem reduces to analysing the sequence of partial products \(\seq[\infty]{\prod_{k=0}^n r(k)}{n=0}\).
In all problem instances where the Threshold Problem is not trivial to determine, 
we employ a classical theorem in analysis (\autoref{thm:productgamma}) that permits us to write the limit of the sequence \(\seq[\infty]{\prod_{k=0}^n r(k)}{n=0}\) as a quotient of two finite products involving the gamma function.
Thus the Threshold Problem for hypergeometric sequences reduces to testing an equality between gamma products.
Our novel approach leverages algebraic and transcendental properties to settle such equality tests.
For example, we frequently employ the algebraic independence of transcendental constants \(\pi\) and \(\eu^{\pi}\) (a celebrated consequence of Nesterenko's work on modular functions \cite{nesterenko1996modular}).

There is a large corpus of research connecting hypergeometric sequences and the gamma function  (often under the guise of infinite product identities~\cite{allouche2015gamma, chamberland2013gamma, dilcher2018gamma}).
This is particularly relevant for our approach (as described above) and sets us apart from previous papers on the Membership Problem in this setting~\cite{kenison2023membership, nosan2022membership}. 
As a nod to the wider appeal of our approach, let us consider two examples from the literature on numeric and symbolic computation.

\begin{example}[{Numerical evaluation of the Kepler--Bouwkamp constant \cite[Section 4]{chamberland2013gamma}}] \label{ex:KB}
This example closely follows work by Chamberland and Straub~\cite{chamberland2013gamma}.
Those authors demonstrate the use of hypergeometric sequences to efficiently approximate certain numerical constants given by infinite products.
One such example is the Kepler--Bouwkamp constant \(\prod_{k=3}^\infty \cos(\pi/k)\).
The convergence of this infinite product is notoriously slow: the error bound between the approximation \(\prod_{k=3}^{10^4} \cos(\pi/k)\) and the Kepler--Bouwkamp constant is \(10^{-4}\).

Let us consider a Pad\'{e} approximation that uses hypergeometric sequences with quadratic parameters. We recall that
the \emph{\([2,2]\)-Pad\'{e} approximant} of a function \(f(x)\) is the rational function \(q(x)/p(x)\) where \(p,q\in\Z[x]\) are quadratic polynomials for which the Maclaurin series of \(q(x)/p(x)\) agrees with that of \(f(x)\) up to order \(4\).
For example, the \([2,2]\)-Pad\'{e} approximant  of \(\cos(x)\) is
	\begin{equation*}
		r_2(x) := \frac{12-5x^2}{12+x^2} = \cos(x) + O(x^6).
	\end{equation*}
Now consider
				\begin{equation*}
					\prod_{k=3}^\infty r_2(\pi/k) = \prod_{k=3}^\infty \frac{12k^2- 5\pi^2}{12k^2 + \pi^2} = \frac{\Gamma(3- \tfrac{\iu}{6}\sqrt{3}\pi)\Gamma(3+\tfrac{\iu}{6}\sqrt{3}\pi)}{\Gamma(3-\frac{1}{6}\sqrt{15}\pi)\Gamma(3+\frac{1}{6}\sqrt{15}\pi)}.
				\end{equation*} 
Here the evaluation as a quotient of two gamma products follows from \autoref{thm:productgamma}.
We pass this evaluation to any modern computer algebra system and determine that  the error between \(\prod_{k=3}^\infty r_2(\pi/k) \) and the Kepler--Bouwkamp constant is bounded by \(10^{-3}\).

Loosely speaking, the \([2,2]\)-Pad\'{e} approximant leverages a hypergeometric sequence with quadratic parameters.\footnote{The quadratics \(12k^2-5\pi^2\) and \(12k^2+\pi^2\) do not have rational coefficients, but in \appref{app:KB} we demonstrate how our approach handles decidability of the Threshold Problem in this example.}  Higher-order approximants will give a closer numerical estimate via hypergeometric sequences with higher-degree parameters.
\end{example}

\begin{example}[Evaluation of a gamma product with cyclotomic parameters]
Here we include a simple (yet concrete) example~\cite[pages 4--6]{borwein2004experimentation} of the difficulty in determining whether a hypergeometric sequence converges to a rational limit.
Consider
\begin{equation} \label{eq:borwein}
	\prod_{k=2}^\infty \frac{k^5 - 1}{k^5 +1} = \frac{2 \cdot \Gamma\bigl(-\omega_{10}\bigr) \Gamma\bigl(\omega_{10}^2\bigr) \Gamma\bigl(-\omega_{10}^3\bigr) \Gamma\bigl(\omega_{10}^4 \bigr)}{ 5 \cdot \Gamma\bigl(\omega_{10}\bigr) \Gamma\bigl(-\omega_{10}^2\bigr) \Gamma\bigl(\omega_{10}^3 \bigr) \Gamma \bigl(-\omega_{10}^4 \bigr)}
\end{equation}
where \(\omega_{10} = \eu^{2\pi \iu/10}\) and, once again, the right-hand side is derived from \autoref{thm:productgamma}.
As noted by the authors of~\cite{borwein2004experimentation}, it is not known whether the limit in \eqref{eq:borwein} is even algebraic.

Whilst the state of the art cannot generally handle the evaluation of expressions given by gamma products, many works in the literature have established identities for restricted classes of products
(a non-exhaustive list includes \cite{chamberland2013gamma, martin2009product, nijenhuis2010gamma, paulsen2019gamma,sandor1989gamma, vidunas2005gamma}).
Our connection to such interests stems from our approach herein: our reduction-step leaves us to determine whether the ratio of two gamma products (as above) is rational. 
\end{example}

\subsection*{Related Work}

\subparagraph*{Membership for Hypergeometric Sequences.} Two recent works consider the Membership Problem for hypergeometric sequences \cite{kenison2023membership, nosan2022membership}.
In both of these works, the authors use \(p\)-adic techniques and divisibility arguments (in stark contrast to the approach herein).
It is worth noting that such techniques seem appropriate only for the Membership Problem and not the Threshold Problem.

The authors of \cite{nosan2022membership} establish decidability of the Membership Problem for the class of hypergeometric sequences with rational parameters.
Closer to our setting, the authors of \cite{kenison2023membership} establish decidability of the Membership Problem for the class of hypergeometric sequences whose polynomial coefficients (as in \eqref{eq:rec}) are both monic and split over a quadratic field.
By comparison to \cite{kenison2023membership}, we establish decidability of not only the Membership Problem, but also the Threshold Problem (\autoref{thm:mainquadratic}) for hypergeometric sequences whose monic polynomial coefficients split over an imaginary quadratic field.
We note our result for the Membership Problem is weaker for sequences whose monic polynomial coefficients split over a real quadratic field: in this setting we are limited to conditional decidability (\autoref{prop:reduction} and \autoref{cor:quadraticSC}).

\subparagraph*{Positivity for P-finite sequences.}
Identities for P-finite sequences are frequent in the literature; however, as noted by Kauers and Pillwein, “in contrast,\ldots almost no algorithms are available for inequalities'' in this setting \cite{kauers2010positive}.
Determining whether the terms of a P-finite sequence are non-negative 
has garnered much attention in recent works (see, for example, \cite{ibrahim2023positivity, kauers2010positive, kenison2020positivity, pillwein2015positivity}).
On the one hand, these works handle higher-order P-finite sequences than the hypergeometric sequences we consider.
On the other hand, the algorithms described in the above studies are restricted in their applicability (placing syntactic restrictions on the polynomial coefficients) and termination is not guaranteed for all initial values.
Indeed, genericity of initial conditions (in the sense that, the growth rate of a recurrence sequence is determined by a positive dominant eigenvalue) is required for the algorithms in \cite{ibrahim2023positivity, kauers2010positive}.
Additionally, determining whether the initial conditions of a given sequence are generic is an open problem (even at low orders) \cite{kenison2020positivity}.

\subparagraph*{Positivity for C-finite sequences.}

It is easily seen that the Threshold Problem for C-finite sequences reduces to the {Positivity Problem} for C-finite sequences. 
Recall that \emph{Positivity} asks to determine whether all terms in a sequence lie above the threshold zero (so the variant of the Threshold Problem where \(t=0\)).
This reduction is straightforward: %
to determine whether \(u_n\ge t\) for each \(n\in\N_0\) 
we can equivalently ask whether \(v_n := u_n - t \ge 0\) for each \(n\in\N_0\).
Observe that \(\seq[\infty]{v_n}{n=0}\) is C-finite since it is given by the difference of two C-finite sequences \(\seq[\infty]{u_n}{n=0}\) and \(\seq[\infty]{t}{n=0}\) and we are done.

   Decidability of the Positivity Problem for C-finite sequences is considered a challenging open problem.
   Further, Positivity and its variants have garnered much research interest in recent works \cite{halava2006positivity, kenison2023positivity, ouaknine2014simple, ouaknine2014ultimate, ouaknine2015termination, ouaknine2014positivity}.
Akin to our focus on restricted classes of hypergeometric sequences herein, the authors of \cite{kenison2023positivity} and \cite{ouaknine2014simple} consider restricted classes of C-finite sequences: both of those works place restrictions on the algebraic properties of the associated recurrence relations.

\subsection*{Structure and Outline}

This paper is structured as follows.
In the next section we gather together relevant preliminary material.
In \cref{sec:maingauss,sec:schanuel}, we establish (un)conditional decidability of the Threshold Problem for classes of hypergeometric sequences.
In one sense, \autoref{sec:maingauss} gives an overview of our approach in the setting of hypergeometric sequences with quadratic parameters.
In \autoref{sec:schanuel}, we introduce the class of polynomials with \autoref{ass:symmetry} (\autoref{ssec:4.1}) and then show that the Threshold Problem for hypergeometric sequences whose monic polynomial coefficients possess \autoref{ass:symmetry} is decidable subject to the truth of Schanuel's conjecture (\autoref{thm:mainclass} in \autoref{ssec:4.2}). 
We make suggestions for future avenues of research in the conclusion (\autoref{sec:conclusion}).
Proofs omitted from the main text are given in \appref{app:turing}.
\appref{app:KB} contains a worked example related to the  Kepler--Bouwkamp (\cref{ex:KB}) and demonstrates an application of Schanuel's conjecture.

\section{Preliminaries} \label{sec:prelim}

\subparagraph*{The Gamma Function.} \label{ssec:gamma}

The approach herein relies on transcendence theory for the gamma function \(\Gamma\) where
	\begin{equation*}
		\Gamma(z) = \int_0^\infty x^{z-1} \eu^{-x}\, dx \quad \text{for \(z\in\mathbb{C}\) with \(\Re(z)>0\).}
	\end{equation*}
It is possible to analytically extend the domain of \(\Gamma\) to the whole complex plane minus the non-positive integers where the function has simple poles.
We briefly recall standard results for the gamma function.
Further details and historical accounts are given in a number of sources (cf.~\cite{andrews1999special,whittaker1996analysis}).

The standard relations for the gamma function give the functional identities:
 the \emph{recurrence} (or \emph{translation}) \emph{property} \(\Gamma(z+1) = z\Gamma(z)\) for \(z\notin \Z\) and the \emph{reflection property} \(\Gamma(z)\Gamma(1-z) = \pi/\sin(\pi z)\) for \(z\notin \Z\).
In the domain of the gamma function, repeated application of the translation property leads to the following `rising factorial' identity.
For \(n\in\{1,2,\ldots\}\), we have
\begin{equation*}
\frac{\Gamma(z+n)}{\Gamma(z)} = z(z+1)\cdots (z+n-1).
\end{equation*}
Similarly, the `falling factorial' identity is given by
\begin{equation*}
\frac{\Gamma(z+1)}{\Gamma(z-n+1)} = z(z-1)\cdots (z-n+1).
\end{equation*}

The next technical lemma is derived from the aforementioned properties of the gamma function.  We employ the notation \(\tfrac{1}{2}\Z\) for the set of integers and half-integers.
\begin{lemma} \label{lem:pairproduct}
Let \(\rho \in \tfrac{1}{2}\Z\).
Suppose that \(w\in\mathbb{C}\) is an algebraic number such that both \(\rho+w\) and \(\rho-w\) lie in the domain of the gamma function and \(w \not\in \tfrac{1}{2}\Z\).
Up to multiplication by an algebraic number, we have the following equalities:
\begin{equation*}
 \Gamma(\rho + w)\Gamma(\rho- w) = 
 	\begin{dcases}
 		\frac{2\pi\iu \eu^{\pi w \iu} }{w (1-\eu^{2\pi w \iu})} & \text{if $\rho$ is an integer, or} \\
 		\frac{2\pi \eu^{\pi w \iu}}{\eu^{2\pi w \iu} + 1} & \text{if $\rho$ is a half-integer}.
 	\end{dcases}
\end{equation*}
\end{lemma}

\begin{proof}
Let us apply the rising and falling factorial identities (as appropriate to the sign of \(\rho\)).
Then, up to multiplication by an algebraic number, we have the following equalities:
\begin{equation*}
 \Gamma(\rho + w)\Gamma(\rho- w) = 
 	\begin{dcases}
 		\Gamma(w)\Gamma(-w) & \text{if $\rho$ is an integer, or} \\
 		\Gamma(1/2+w)\Gamma(1/2-w) & \text{if $\rho$ is a half-integer}.
 	\end{dcases}
\end{equation*}
Consider the first of the two cases above.
The reflection and recurrence formulas lead to
	\begin{equation*}
		\Gamma(w)\Gamma(-w) = \frac{\Gamma(w)\Gamma(1-w)}{-w}
		= \frac{\pi}{-w \sin(\pi w)}
		=  -\frac{2\pi\iu}{w (\eu^{\pi w \iu} - \eu^{-\pi w \iu})}.
	\end{equation*}
For the second case, we employ the cosine variant of Euler's reflection formula to obtain 
\begin{equation*}
	\Gamma(1/2 + w)\Gamma(1/2 - w) = \frac{\pi}{\cos(\pi w)} = \frac{2\pi}{\eu^{\pi w \iu} + \eu^{-\pi w \iu}}.
\end{equation*}
The equalities in the statement of the lemma quickly follow.
\end{proof}

\subparagraph*{Decidability and Reduction Results.}

Recall that a rational-valued sequence \(\seq[\infty]{u_n}{n=0}\) is \emph{hypergeometric} if it satisfies a first-order recurrence relation of the form \eqref{eq:rec} with polynomial coefficients \(p,q\in\Z[x]\).
Due to space restrictions, we omit the proofs of \autoref{lem:MPreduction}, \autoref{prop:reductiontogamma2}, and \autoref{prop:reduction} from the main text.
Each proof is included in \appref{app:turing}.

The following straightforward lemma appears in previous works \cite{kenison2020positivity,nosan2022membership}.
\begin{restatable}{lemma}{lemMPreduction} \label{lem:MPreduction}
	Consider the class of hypergeometric sequences \(\seq{u_n}{n}\) whose shift quotients \(r(n)\)
	either diverge to \(\pm\infty\) or converge to a limit \(\ell\) with \(|\ell|\neq 1\).
	For this class, the Membership and Threshold Problems are both decidable.
\end{restatable}

Thus to decide the Membership and Threshold Problems for hypergeometric sequences, we need only consider the sequences whose  shift quotients \(r(n)\) converge to \(\pm1\) as \(n\to\infty\).
We say an infinite product \(\prod_{k=0}^\infty r(k)\) \emph{converges} if the sequence of partial products converges to a finite non-zero limit (otherwise the product is said to diverge).
Recall the following classical theorem (\cite[\S12]{whittaker1996analysis} and \cite{chamberland2013gamma}).
\begin{theorem} \label{thm:productgamma}
 	Consider the rational function
 		\begin{equation*}
 					r(k) := \frac{c(k+\alpha_1) \cdots (k+\alpha_m)}{(k+\beta_1) \cdots (k+\beta_{m'})}
 		\end{equation*}
 		where we suppose that each \(\alpha_1,\ldots, \alpha_m, \beta_1,\ldots, \beta_{m'}\) is a complex number that is neither zero nor a negative integer.
 		The infinite product \(\prod_{k=0}^\infty r(k)\) converges to a finite non-zero limit only if \(c=1\), \(m=m'\), and \(\sum_j \alpha_j = \sum_j \beta_j\).
 		Further, the value of the limit is given by
 			\begin{equation*}
 			\prod_{k=0}^\infty r(k) = \prod_{j=1}^m \frac{\Gamma(\beta_j)}{\Gamma(\alpha_j)}.
 			\end{equation*}
\end{theorem}

With \autoref{thm:productgamma} in mind, it is useful to introduce the following terminology for shift quotients.
We call a rational function \(r(k)\) (as above) \emph{harmonious}  if \(r(k)\) satisfies the assumptions \(c=1\), \(m=m'\), and \(\sum_j \alpha_j = \sum_j \beta_j\).
From \autoref{thm:productgamma}, it is immediately apparent that a hypergeometric sequence \(\seq{u_n}{n}\) with shift quotient \(r\) converges to a finite non-zero limit only if \(r\) is harmonious.

Related to the assumptions in \autoref{thm:productgamma} (and \autoref{prop:reductiontogamma2} below), when considering Membership and Threshold we can assume without loss of generality that the roots \(\alpha_1,\ldots, \alpha_m\) of the coefficient \(q\) (as in \eqref{eq:rec}) are neither zero nor negative integers. For otherwise, a hypergeometric sequence eventually hits zero and is identically zero thereafter.
In \cite{nosan2022membership}, Nosan et al.\ establish \cref{prop:reductiontogamma2,prop:reduction} for rational parameters.
The proof of \autoref{prop:reductiontogamma2} given in \appref{app:turing} is all but identical to the proof of Proposition 2 in \cite{nosan2022membership}.
\begin{restatable}{proposition}{propreductiontogamma} \label{prop:reductiontogamma2}
Let \(\seq{u_n}{n}\) be a hypergeometric sequence whose shift quotient is given by a ratio of two polynomials with real coefficients.
For such sequences, the Membership and Threshold Problems are both Turing-reducible to the following decision problem.
Given $d\in\N$, $\alpha_1,\ldots,\alpha_d\in\mathbb{C}\setminus\Z_{<0}$ (the roots of some \(P(x)\in\R[x]\)), and $\beta_1,\ldots,\beta_d\in\mathbb{C}\setminus\Z_{<0}$ (the roots of some \(Q(x)\in\R[x]\)),
    determine whether
    \begin{equation*}\label{eq:monomialgamma2}
        \frac{\Gamma(\beta_1)\cdots\Gamma(\beta_d)}{\Gamma(\alpha_1)\cdots\Gamma(\alpha_d)} = t
    \end{equation*}
for \(t\in\Q\setminus 0\).
\end{restatable}
To be absolutely clear, the fact that \(t\in\Q\) is non-zero in \autoref{prop:reductiontogamma2} follows directly from the infinite product in \autoref{thm:productgamma} converging to a non-zero limit.

Our (un)conditional decidability results for the Membership Problem follow from the next proposition.  This proposition can be deduced from the work in \cite{nosan2022membership} (a straightforward proof is given in \appref{app:turing}).

\begin{restatable}{proposition}{propreduction} \label{prop:reduction}
For hypergeometric sequences, decidability of the Membership Problem Turing-reduces to that of the Threshold Problem.
\end{restatable}

\subparagraph*{Number Fields.}

We recall standard results for quadratic fields below (cf.\ \cite[Chapter 3]{stewart2016algebraic}).
A number field \(K\) is \emph{quadratic} if \([K:\Q]=2\).
A field \(K\) is quadratic if and only if there is a square-free integer \(d\) such that \(K = \Q(\sqrt{d})\).
Further, a quadratic field \(\Q(\sqrt{d})\) is \emph{imaginary} if \(d<0\).
\begin{theorem} \label{thm:quadratic}
Suppose that \(d\in\Z\) is  square-free.
Then the algebraic integers of \(\Q(\sqrt{d})\) are given by \(\Z[\sqrt{d}]\) if \(d \not\equiv 1 \pmod{4}\)  or \(\Z[1/2 + \sqrt{d}/2]\) if \(d\equiv 1 \pmod{4}\).
\end{theorem}

We include the following straightforward lemma for ease of reference. %
\begin{lemma} \label{lem:alg2rat}
Let \(\mathbb{L}/\mathbb{Q}\) be a finite Galois extension.  Suppose that \({\mathcal{P}}\in\mathbb{L}(X_1,\ldots, X_m)\) is a polynomial such that \({\mathcal{P}}(s_1,\ldots, s_m)=0\) with \((s_1,\ldots, s_m)\in\mathbb{C}^m\).
Then there is a polynomial \(\mathcal{Q} \in \mathbb{Q}(X_1,\dots, X_m)\) such that \(\mathcal{Q}(s_1,\ldots, s_m)=0\).
\end{lemma}
\begin{proof}
Let \({\mathcal{P}} = \sum_{(t_1,\ldots, t_m)} c_{(t_1,\ldots, t_m)} X_1^{t_1}X_2^{t_2} \cdots X_m^{t_m} \) and for each \(\sigma\in G\) (the Galois group of \(\mathbb{L}/\mathbb{Q}\)) let 
\begin{equation*}
\sigma({\mathcal{P}}) = \sum_{(t_1,\ldots, t_m)} \sigma(c_{(t_1,\ldots, t_m)}) X_1^{t_1}X_2^{t_2} \cdots X_m^{t_m}.
\end{equation*}

Let \(\mathcal{Q} = \operatorname{N}_{\mathbb{L}/\mathbb{Q}}({\mathcal{P}}) := \prod_{\sigma\in G} \sigma({\mathcal{P}})\).
It is clear that each of the coefficients of the polynomial \(\mathcal{Q}\) is rational since the coefficients are invariant under the action of the group \(G\).
Further,
	\begin{equation*}
		\mathcal{Q}(s_1,\ldots, s_m) = {\mathcal{P}}(s_1,\ldots, s_m) \prod_{\sigma\in G\setminus{e_G}} \sigma({\mathcal{P}}) (s_1,\ldots, s_m) = 0,
	\end{equation*}
as desired.
\end{proof}

\subparagraph*{Transcendental Number Theory.}

The transcendence degree of a field extension is a measure of the size of the extension.
In fact, for finitely generated extensions of \(\mathbb{L}/\Q\) (such as those that we consider), the transcendence degree indicates the largest cardinality of an algebraically independent subset of \(\mathbb{L}\) over \(\Q\).
For a field extension \(\mathbb{L}/\Q\), a subset \(\{\xi_1,\ldots, \xi_n\}\subset \mathbb{L}\) is \emph{algebraically independent} over \(\Q\) if for each polynomial 
	\(	P(X_1,\ldots, X_n)\in\Q[X_1,\ldots, X_n] \)
we have that \(P(\xi_1,\ldots, \xi_n)=0\) only if \(P\) is identically zero.

It is useful to recall  the Gelfond--Schneider Theorem that establishes the transcendentality of \(\alpha^\beta\) for algebraic numbers \(\alpha\) and \(\beta\) except for the cases where \(\alpha=0,1\) or \(\beta\) is rational.

Schanuel's conjecture is a unifying prediction in transcendental number theory.
If Schanuel's conjecture is true, then it generalises several of the principal
results in transcendental number theory such as: the Gelfond--Schneider Theorem, the Lindemann--Weierstrass Theorem, and
Baker’s theorem (cf.\ \cite{lang1966introduction, baker1975, Waldschmidt2006}).
The conjecture makes the following prediction: for \(\xi_1, \ldots, \xi_n\) rationally linearly independent complex numbers, there is a subset of 
		\(\{\xi_1, \ldots, \xi_n, \eu^{\xi_1}, \ldots, \eu^{\xi_n}\}\)
of size at least \(n\) that is  algebraically independent over \(\Q\).

\begin{conjecture}[Schanuel]
	Suppose that \(\xi_1, \ldots, \xi_n\in\mathbb{C}\) are linearly independent over the rationals \(\Q\).
	Then the transcendence degree of the field extension 
		\( \Q(\xi_1,\ldots,  \xi_n, \eu^{\xi_1}, \ldots, \eu^{\xi_n})\)
over \(\Q\) is at least \(n\).
\end{conjecture}

\section{Hypergeometric Sequences with Quadratic Parameters} \label{sec:maingauss}

As an appetiser to the proofs of \cref{thm:mainquadratic,thm:mainclass},
we introduce our approach by first establishing decidability of the Threshold Problem for hypergeometric sequences with Gaussian integer parameters (\autoref{prop:maingauss} below).
We also include a worked example in \autoref{ex:wkexample}.
Recall that the \emph{Gaussian integers} \(\Z[\iu]\) are those complex numbers of the form \(a+b\iu\) for which \(a,b\in\Z\).

\begin{proposition} \label{prop:maingauss}
The  Threshold Problem for hypergeometric sequences whose polynomial coefficients are monic and split over \(\Q(\iu)\) are decidable.
\end{proposition}

\begin{proof}
By \autoref{thm:productgamma} and \autoref{prop:reductiontogamma2}, we need only consider hypergeometric sequences with harmonious shift quotients.
Since the polynomial coefficients \(p,q\in\Z[x]\) are monic, the roots and poles of the shift quotient are integers in \(\Q(i)\); that is to say, they lie in \(\Z[i]\). %
It follows from \autoref{lem:pairproduct} that each such instance \((\seq{u_n}{n}, t)\) of the Threshold Problem reduces to testing an equality of the form
	\begin{equation} \label{eq:gausstest}
		\theta \pi^{\ell} \prod_{m}  (\eu^{b_m \pi} - \eu^{-b_m \pi})^{\varepsilon_m} = t.
	\end{equation}
Here \(\theta\) is rational and non-zero, \(\ell\in\Z\), each pair \(\Gamma(a_m+b_m\iu)\Gamma(a_m-b_m\iu)\) from \autoref{prop:reductiontogamma2} contributes a term \( (\eu^{b_m \pi} - \eu^{-b_m \pi})^{\varepsilon_m}\) in the finite product, and \(\varepsilon_m = \pm 1\).

We break the remainder of the proof into several subcases.
Without loss of generality, we can assume that not all the roots and poles of \(r\) are rational integers, for otherwise testing \eqref{eq:gausstest}     reduces to the decidable task of testing equality between two rational numbers.

We continue under the assumption that not all the roots and poles of \(r\) are rational integers.
Let us now consider the product in \eqref{eq:gausstest}. %
Up to multiplication by a rational, we can write the left-hand side of \eqref{eq:gausstest} in the form
	\begin{equation*}
	 \theta \pi^{\ell} \prod_{m} (\eu^{b_m \pi} - \eu^{-b_m \pi})^{\varepsilon_m} 
		 = \theta \pi^{\ell} \frac{ f(\eu^{\pi}) }{ g(\eu^{\pi}) }
	\end{equation*}
where \(f,g\in\Q[X]\) are non-trivial polynomials.
Observe that \(\eu^{\pi} = (\eu^{\pi\iu})^{-\iu} = (-1)^{-\iu}\);
thus, by the Gelfond--Schneider theorem, \(\eu^{\pi}\) is transcendental.
We break the remainder of the proof into two cases depending on the rationality of \(f(\eu^\pi)/g(\eu^\pi)\).

Suppose that \(f(\eu^\pi)/g(\eu^\pi)\in\Q\).
There are two further subcases to consider:
if \(\ell=0\), then, once again, the equality test \eqref{eq:gausstest} reduces to deciding whether two rationals are equal; and
if \(\ell\neq 0\), then the equality test \eqref{eq:gausstest} reduces to testing whether \(\pi^\ell\) is equal to a given rational number, which cannot hold for then \(\pi\) is necessarily algebraic.

All that remains is to consider the case where \(f(\eu^\pi)/g(\eu^\pi)\not\in\Q\), which we again split into two subcases.
If \(\ell=0\), then it is trivial to see that \eqref{eq:gausstest} cannot hold as the right-hand side is rational.
If \(\ell\neq 0\) and we assume, for a contradiction, that \eqref{eq:gausstest} holds, then
				a simple rearrangement of \eqref{eq:gausstest} shows that there is a non-trivial polynomial \(\mathcal{P}\in\Q[X,Y]\) such that \(\mathcal{P}(\pi,\eu^\pi)=0\).
				This contradicts Nesterenko's theorem \cite{nesterenko1996modular} that \(\pi\) and \(\eu^\pi\) are algebraically independent.
We have dispatched each of the subcases and conclude the desired result.
\end{proof}

\begin{example} \label{ex:wkexample}
Suppose that \(\seq[\infty]{u_n}{n=0}\) is the hypergeometric sequence defined by
	\begin{equation*}
		u_n = \frac{n^2-4n+5}{n^2-4n+13} u_{n-1} \text{ with } u_0=1.
	\end{equation*}
For the Threshold Problem, let us consider the problem instance \((\seq{u_n}{n},t)\) with \(t\in\Q\).
First, we evaluate the associated infinite product:
\begin{equation*}
	\prod_{k=0}^\infty \frac{k^2-4k+5}{k^2-4k+13} = \frac{\Gamma(-2-3\iu)\Gamma(-2+3\iu)}{\Gamma(-2-\iu)\Gamma(-2+\iu)} 
	= \frac{\sinh(\pi)}{39\sinh(3\pi)} = \frac{\eu^{3\pi}(\eu^{2\pi} -1)}{39 \eu^{\pi} (\eu^{6\pi} -1)}.
\end{equation*}
Second, as directed by the proof of \autoref{prop:maingauss}, 
decidability of the Threshold Problem in this instance reduces to determining whether the following  equality holds: 
	\begin{equation*}
		\frac{\eu^{3\pi}(\eu^{2\pi} -1)}{39 \eu^{\pi} (\eu^{6\pi} -1)}=t.
	\end{equation*} 

A simple rearrangement shows that if the above equality holds, then there is a non-trivial polynomial \(p\in\Q[X]\) such that \(p(\eu^{\pi})=0\), from which we deduce that \(\eu^{\pi}\) is algebraic.
However, by the Gelfond--Schneider theorem, \(\eu^{\pi}\) is transcendental.
We have reached a contradiction and deduce that the aforementioned equality cannot hold.

Finally, as described in \autoref{prop:reductiontogamma2}, the Threshold Problem reduces to an exhaustive search of a computable number of initial terms in the sequence \(\seq{u_n}{n}\).
As an aside, by \autoref{prop:reduction}, we also obtain decidability of the Membership Problem for problem instances \((\seq{u_n}{n},t)\).
\end{example}

\begin{remark} \label{rem:target}
	Subject to appropriate changes and by employing \autoref{lem:alg2rat}, we can extend the result in \autoref{prop:maingauss} from instances  of the Membership and Threshold Problems \((\seq{u_n}{n}, t)\) with \(u_0, t\in\Q\) to problem instances with \(u_0, t\in\mathbb{L}(\pi, \eu^{\pi})\) where \(\mathbb{L}\) is any finite Galois extension of \(\Q\).
	This extension similarly holds for \autoref{thm:mainquadratic} (below).
\end{remark}

\begin{remark} \label{rem:famousrings}
Analogous decision procedures to the proof of \autoref{prop:maingauss} also hold for other famous rings of integers: the Eisenstein, Kummer, and Kleinian integers.
Recall that
the \emph{Eisenstein integers} are the elements of  \(\Z[\zeta_3] = \{a+b\zeta_3 : a,b\in\Z\}\) where \(\zeta_3 := \eu^{2\pi\iu/3}\).
Similarly, the \emph{Kummer integers} are the elements of \(\Z[\sqrt{-5}] = \{a+b\sqrt{-5} : a,b\in\Z\}\).
Finally, the \emph{Kleinian integers} are the elements of \(\Z[\mu] = \{a+b\mu : a,b\in\Z\}\) where \(\mu = -1/2 + \sqrt{-7}/2\).
\end{remark}

The claims for decidability in \autoref{rem:famousrings}, follow from the next theorem.
We establish decidability of the Threshold Problem for hypergeometric sequences whose parameters are drawn from the ring of integers of an imaginary quadratic number field.
\begin{theorem} \label{thm:mainquadratic}
The Threshold Problem for hypergeometric sequences whose polynomial coefficients are monic and split over an imaginary quadratic number field is decidable.
\end{theorem}
\begin{proof}

Mutatis mutandis, the proof of \autoref{thm:mainquadratic} follows the approach in \autoref{prop:maingauss}.
For the sake of brevity, we shall indicate only the major changes to \autoref{prop:maingauss} here.
Consider the ring of integers of an imaginary quadratic field \(\Q(\sqrt{d})\) where \(-d\in\N\) is square-free.
By \autoref{thm:quadratic}, there are two cases to consider: first, when \(d\not\equiv 1 \pmod{4}\) and second, when \(d\equiv 1 \pmod{4}\).

We concentrate on the changes to the proof of \autoref{prop:maingauss} when \(d\not\equiv 1 \pmod{4}\).
Like before, we can use the recurrence formula to write \(\Gamma(a+b\sqrt{d}) = \theta \Gamma(b\sqrt{d})\) where \(\theta\in\N\).
Thus all that remains is to evaluate products \(\Gamma(b\sqrt{d})\Gamma(-b\sqrt{d})\) of conjugate elements.
By the reflection formula, we have
\begin{equation*}
	\Gamma(b\sqrt{d})\Gamma(-b\sqrt{d}) = - \frac{\pi}{b\sqrt{d}\sin(b\pi\sqrt{d})}
	= \frac{2\pi}{b\sqrt{-d}(\eu^{\pi b\sqrt{-d}}-\eu^{-\pi b\sqrt{-d}} )}.
\end{equation*}
The important update is the product in \eqref{eq:gausstest}. %
In our new setting, the product takes the form 
	\begin{equation*}
		\prod_{m} (\eu^{\pi b\sqrt{-d}}-\eu^{-\pi b\sqrt{-d}})^{\varepsilon_m}.
	\end{equation*}
Observe that \(\eu^{\pi\sqrt{-d}}\) is transcendental (once again by Gelfond--Schneider) and that for each \(-d\in\N\) the numbers \(\pi\) and \(\eu^{\pi\sqrt{-d}}\) are algebraically  independent over \(\Q\) \cite[Corollary 6]{nesterenko1996modular}.
The rest of the proof in this case follows as before.

In the second case where \(d\equiv 1 \pmod{4}\) we must additionally deal with contributions of the form \(\Gamma(b/2 + b\sqrt{d}/2)\Gamma(b/2 - b\sqrt{d}/2)\).
This setting introduces cases where \(2 \nmid b\), which we resolve by repeated application of the recurrence formula and the cosine variant of Euler's reflection formula.
Indeed, we have
\begin{equation*}
	\Gamma(1/2 + b\sqrt{d}/2)\Gamma(1/2 - b\sqrt{d}/2) = \frac{\pi}{\cos(\pi b\sqrt{d}/2)} 
		= \frac{2\pi}{\eu^{\pi b\sqrt{-d}/2} + \eu^{-\pi b \sqrt{-d}/2}},
\end{equation*}
and so we can construct an updated version of the product in \eqref{eq:gausstest}.
This update and analogous arguments for the transcendental properties of \(\eu^{\pi\sqrt{-d}/2}\) let us conclude decidability in this case too.
\end{proof}

\section{Conditional Decidability Subject to Schanuel's conjecture} \label{sec:schanuel}

In this section we shall give a generalisation (\autoref{thm:mainclass}) of the decidability results in \autoref{sec:maingauss}.
The result applies to a strictly larger class of hypergeometric sequences; however, we sacrifice unconditional decidability.
Briefly, termination of the decidability procedure in \autoref{thm:mainclass} is dependent on the truth of Schanuel's conjecture (further details are given in \autoref{rem:MacintyreWilkie}).
The motivation for studying reachability problems for this larger class arises from interest in the literature for limits of hypergeometric sequences with unnested radical and cyclotomic parameters (see \autoref{ssec:4.2}).

Let us first describe the class of hypergeometric sequences for which we establish conditional decidability of the Membership and Threshold Problems.
This will require us to introduce some notations and \autoref{ass:symmetry} below.

Consider \(\mathcal{R}_f\) the multiset of roots of a given monic polynomial \(f\in\Z[x]\) and \(\mathcal{V}_f\) the multiset of irrational roots of \(f\).
We define the graph \(\mathcal{G}_f :=(\mathcal{V}_f, \mathcal{E}_f)\)
with vertex set \(\mathcal{V}_f\) and edge set \(\mathcal{E}_f\).
Here \(\mathcal{E}_f\) is encoded using an adjacency matrix \(A\) (whose rows and columns are indexed by \(\mathcal{V}_f\)) as follows.
For distinct \(u,v\in\mathcal{V}_f\),
let \(A(u,v)=1\) if \(u = \rho + w\) and \(v = \rho - w\) for some \(\rho \in \tfrac{1}{2}\Z\) and \(w\) an algebraic number not in \(\tfrac{1}{2}\Z\). Otherwise, let \(A(u,v) = 0\).

\begin{assumption} \label{ass:symmetry}

We say that \(f\) has \autoref{ass:symmetry} if \(\mathcal{G}_f\) admits a perfect matching.
\end{assumption}

The motivation for introducing \autoref{ass:symmetry} is as follows:
in combination with \autoref{lem:pairproduct} and subject to the truth of Schanuel's conjecture, we can circumvent certain hard problems concerning the evaluation of gamma products if the input parameters possess certain symmetries.
The main result of this section is a conditional decidability result for the class of hypergeometric sequences whose polynomial coefficients are monic and possess \autoref{ass:symmetry}.

\begin{restatable}{theorem}{thmmainclass} \label{thm:mainclass}
The Threshold Problem for hypergeometric sequences whose polynomial coefficients (as in \eqref{eq:rec}) are both monic and have \autoref{ass:symmetry} is decidable subject to the truth of Schanuel's conjecture.
\end{restatable}

The remainder of this section is structured as follows.
In \autoref{ssec:4.1} we list classes of polynomials with \autoref{ass:symmetry} and highlight straightforward corollaries (\cref{cor:quadraticSC,cor:minimalrationalSC,cor:unnested}) of \autoref{thm:mainclass}.
In \autoref{ssec:4.2} we prove \autoref{thm:mainclass}.

\subsection{Polynomials with \autoref{ass:symmetry}} \label{ssec:4.1}

\subparagraph*{Even Polynomials.}
It is immediate that even polynomials with at least one irrational root possess \autoref{ass:symmetry}.
Suppose that \(f\in\Z[x]\) is a monic even polynomial.  Then \(f(-x)=f(x)\) and so \(\mathcal{G}_f\) admits a perfect matching, as desired.  More generally, we note that a horizontal translation of such a polynomial, say \(\tilde{f}\), for which \(\tilde{f}(\rho-x)=\tilde{f}(\rho+x)\) where \(\rho\in\tfrac{1}{2}\Z\), also has \autoref{ass:symmetry}.

\subparagraph*{Quadratic Polynomials.}
Every irreducible monic quadratic polynomial in \(\Z[x]\) possesses \autoref{ass:symmetry}.  This observation is straightforward: consider an irreducible monic quadratic polynomial \(x^2 + bx + c\in\Z[x]\).  The roots of said quadratic satisfy \(-\tfrac{b}{2} \pm \tfrac{\sqrt{b^2-4c}}{2}\).  We note that \(-\tfrac{b}{2}\in\tfrac{1}{2}\Z\) and \(b^2-4c\neq 0\).

The following corollary is a straightforward consequence of \autoref{thm:mainclass}: we have conditional decidability of the  Threshold Problem for hypergeometric sequences whose parameters are drawn from the rings of integers of quadratic number fields.
\begin{corollary} \label{cor:quadraticSC}
The Threshold Problem for hypergeometric sequences whose polynomial coefficients are monic with irreducible factors that are either linear or quadratic is decidable subject to the truth of Schanuel's conjecture.
\end{corollary}
We note that the assumption in \autoref{cor:quadraticSC} permits us to draw sequence parameters from the integers of any number of quadratic fields in order to establish conditional decidability of both the Threshold and Membership Problems.
This is in stark contrast to the work in \cite{kenison2023membership} that establishes decidability of Membership for hypergeometric sequences whose parameters are drawn from the integers of a single quadratic field.

\subparagraph*{Algebraic Numbers with Rational Real Part.}
Consider the class \(\mathcal{C}\) of monic irreducible polynomials in \(\Q[x]\) that possess a root with rational real part.
Trivially, a linear polynomial with rational coefficients is always a member of \(\mathcal{C}\) since the single root of the polynomial is rational.
 It is straightforward to see that an irreducible quadratic polynomial \(x^2+bx+c\in\Q[x]\) is in \(\mathcal{C}\) if and  only if \(b^2-4c<0\).
Dilcher, Noble, and Smyth \cite{dilcher2011minimal} achieve a complete classification of the class \(\mathcal{C}\) with the following result.
\begin{theorem}[{\cite[Theorem 1]{dilcher2011minimal}}] \label{thm:minimalrational}
	Let \(f\) be a polynomial of degree at least three.
	Then \(f\in\mathcal{C}\) if and only if \(f(x) = g((x-\rho)^2)\) for some \(\rho\in\Q\) and monic irreducible \(g\in\Q[X]\) that has a negative real root.
	In this case, \(f\) has a root with a rational real part \(\rho\).
\end{theorem}
The following corollary is key to understanding the connection to \autoref{ass:symmetry} and demonstrates the 2-fold rotational symmetry of the set of roots of a polynomial in \(\mathcal{C}\).
\begin{corollary}[{\cite[Corollary 2]{dilcher2011minimal}}] \label{cor:minimalrational}
 Suppose that \(f\in\mathcal{C}\) has degree at least three.
 The roots of \(f\) that have rational real part have the same real part \(\rho\).
 Further, we have \(f(\rho-x) = f(\rho+x)\).
\end{corollary}

The trivial lemma below shows that if an algebraic integer has rational real part, then said real part lies in \(\tfrac{1}{2}\Z\).

\begin{lemma} \label{lem:minimalrational}
Let \(\alpha\) be an algebraic integer with \(\Re(\alpha)\in\Q\).
Then \(\Re(\alpha)\in\tfrac{1}{2}\Z\).
\end{lemma}
\begin{proof}
Since
\(
		\Re(\alpha) = \tfrac{1}{2}(\alpha + \oalpha)
\),
we deduce that \(\alpha + \oalpha \in\Q\).
Observe that both \(\alpha\) and \(\oalpha\) are algebraic integers and so \(\alpha + \oalpha\) is too due to the closure of the ring of algebraic integers.
We deduce that \(\alpha + \oalpha\in\Z\), from which the desired result follows.
\end{proof}

When we combine the result in \autoref{thm:mainclass} with the observations in \cref{cor:minimalrational,lem:minimalrational} we obtain the following corollary.

\begin{corollary} \label{cor:minimalrationalSC}
The Threshold Problem for hypergeometric sequences whose polynomial coefficients are monic with irreducible factors in \(\mathcal{C}\) is decidable subject to the truth of Schanuel's conjecture.
\end{corollary}

\subparagraph*{Unnested Radicals and Cyclotomic Polynomials.}
We now highlight the class of hypergeometric sequences with parameters determined by unnested radicals and cyclotomic polynomials.
The limits of such sequences are considered in both \cite[pp.\ 753--757]{prudnikov1986integrals} and \cite{dilcher2018gamma}.
In the sequel we use \(\Phi_d\) to denote the \(d\)th cyclotomic polynomial.
We state and prove the following corollary of \autoref{thm:mainclass}.

\begin{corollary} \label{cor:unnested}
 The  Threshold Problem for hypergeometric sequences whose polynomial coefficients have irreducible factors of the form \(x^d - a\) with \(d\in 2\Z\), or \(\Phi_d\) with \(d\in 4\Z\)
 is decidable subject to the truth of Schanuel's conjecture.
\end{corollary}

\begin{proof}
In light of \autoref{thm:mainclass}, it is sufficient to prove that non-linear  polynomials of the form \(x^d - a\) with \(d\in 2\Z\), or \(\Phi_d\) with \(d\in 4\Z\) have \autoref{ass:symmetry}.

We first consider irreducible factors of the form  \(x^d - a \in\Z[x]\).
The roots of \(x^d -a\) are unnested radicals of the form \(\sqrt[m]{a} \omega_m^j\) for \(j\in\{0,\ldots, m-1\}\) where \(\omega_d:=\eu^{2\pi\iu/d}\).
Recall that \(x^d - a \in\Z[x]\) is irreducible if \(a\) is not the \(N\)th power of an element of \(\Q\) for some \(N>1\) with \(N\mid d\) (cf.\ \cite{mordell1953linear}).
When \(d\) is even, it follows that \(x^d -a\) is even and so possesses  \autoref{ass:symmetry}.

The cyclotomic polynomial \(\Phi_d\) is irreducible and its roots are the primitive \(d\)th roots of unity.
Under the assumption that \(d\) is a multiple of four, it is straightforward to show that \(\Phi_d\) is even.
It follows that for such \(d\), \(\Phi_d\) has \autoref{ass:symmetry}.
\end{proof}

We note with our approach we cannot lift the additional assumption that \(d\in 4\Z\) for cyclotomic factors:
the set of primitive 18th roots of unity \(\{\omega_{18}, \omega_{18}^5, \omega_{18}^7, \omega_{18}^{11}, \omega_{18}^{13}, \omega_{18}^{17}\}\) show that \(\Phi_{18}\) does not have \autoref{ass:symmetry}.

\subsection{Proof of \autoref{thm:mainclass}} \label{ssec:4.2}
We now prove our main result.
A worked example, demonstrating our approach, is given in \appref{app:KB}.
\thmmainclass*
\begin{proof}

Let \(\seq[\infty]{u_n}{n=0}\) be a hypergeometric sequence whose polynomial coefficients are both monic and satisfy \autoref{ass:symmetry}.
We assume without loss of generality that the associated shift quotient \(r(x):= q(x)/p(x)\) is harmonious.
As previously noted, each instance of the  Threshold Problem \((\seq{u_n}{n}, t)\) with \(t\in\Q\) reduces to checking an equality of the form
\begin{equation} \label{eq:gammaprod}
 \Gamma(\beta_1)\cdots \Gamma(\beta_d) = t \Gamma(\alpha_1)\cdots \Gamma(\alpha_d)
\end{equation}
where \(\{\alpha_1,\ldots, \alpha_d\} =: \mathcal{R}_p\) and \(\{\beta_1,\ldots, \beta_d\} =: \mathcal{R}_q\)
are the multisets of the roots of the respective polynomial coefficients \(p\) and \(q\). 
The proof is split into two parts: a reduction to an equality testing problem and a proof of decidability subject to the truth of Schanuel's conjecture.

\subparagraph*{Reduction to Equality Testing.}
Since \(pq\) is a monic polynomial, the rational elements of the multisets \(\mathcal{R}_p\) and \(\mathcal{R}_q\)
are rational integers (so we can assume their contributions are absorbed into the rational parameter \(t\) in \eqref{eq:gammaprod}).
Further, under \autoref{ass:symmetry} there are perfect matchings on both of the graphs \(\mathcal{G}_p\) and \(\mathcal{G}_q\), which we denote by \(\mathcal{M}_p\) and \(\mathcal{M}_q\) respectively.
Indeed, we recall that in \(\mathcal{G}_p\) the edges in \(\mathcal{M}_p\) are of the form \(e(\alpha_-, \alpha_+)\) such that \(\alpha_\pm = \rho_e \pm w_e\) where \(\rho_e \in \tfrac{1}{2}\Z\) and \(w_e\) is an algebraic number not in \(\tfrac{1}{2}\Z\) (and similarly for \(\mathcal{G}_q\) and \(\mathcal{M}_q\)).
In this setting, we repeatedly apply the recurrence formula and absorb the resulting algebraic factors into a single term \(\theta\) in order to rewrite \eqref{eq:gammaprod} as
	\begin{equation} \label{eq:gammaprod2} 
			\prod_{i\in\mathcal{M}_q} \Gamma(w_i)\Gamma(-w_i) = \theta \prod_{j\in\mathcal{M}_p} \Gamma(w_j)\Gamma(-w_j).
	\end{equation}

Consider the set of algebraic numbers \(\{w_1,\ldots, w_M\}\) determined by the parameters in~\eqref{eq:gammaprod2}.
We denote by \(S':=\{s_1', \ldots, s_m'\}\) a maximal subset of \(\{w_1, \ldots, w_M\}\) for which the elements of \(\{\pi, \pi\iu\} \cup \pi S'\) are \(\Q\)-linearly independent (here \(\pi S':=\{\pi s_1', \ldots, \pi s_m'\}\)).
Then for each \(k\in\{1,\ldots, M\}\), write \(w_k\) as a \(\Q\)-linear sum of elements in \(\{s_1',\ldots, s_m'\}\) so that
	\begin{equation*}
		w_k = \frac{x_{k1}}{y_{k1}} s_1' + \cdots + \frac{x_{km}}{y_{km}} s_m'.
	\end{equation*}
We define \(s_j := s_j' / \operatorname{lcm}(y_{1j}, y_{2j}, \ldots, y_{Mj})\) for each \(j\in\{1,\ldots, m\}\).
Now we can write each \(w_k\in\{w_1,\ldots, w_M\}\) as a \(\Z\)-linear sum of elements in the normalised set \(S:= \{s_1, \ldots, s_m\}\).

We apply \autoref{lem:pairproduct} to \eqref{eq:gammaprod2} and, by the preceding paragraph, determine that the problem of 
testing the equality in \eqref{eq:gammaprod2} reduces to that of determining whether a certain non-trivial polynomial with coefficients in \(\Q(\theta)\) vanishes at a given point (this is analogous to the process in \autoref{prop:maingauss} and \autoref{ex:wkexample}).
More specifically, we want to test whether a given non-trivial polynomial  \(P\in\Q(\theta)[X_1,\ldots, X_{4m+4}]\) satisfies
	\begin{equation} \label{eq:polytest}
		P \bigl( \pi, \pi\iu, \pi S, \pi S\iu, \eu^{\pi}, \eu^{\pi \iu}, \eu^{\pi S}, \eu^{\pi S\iu} \bigr) = 0.
	\end{equation}
Here \(\pi S\iu := \{\pi s_1\iu, \ldots, \pi s_m \iu\}\), \(\eu^{\pi S} := \{\eu^{\pi s_1}, \ldots, \eu^{\pi s_m}\}\), and likewise for \(\eu^{\pi S\iu}\).  
Note that we need only consider a polynomial in \(4m+4\) variables as the parameters \(\{w_1,\ldots, w_M\}\) of our problem instance are given by \(\Z\)-linear combinations of the elements of \(S\cup S\iu\).
We claim that  the equality in \eqref{eq:polytest} cannot hold if Schanuel's conjecture is true.
We prove this claim below.

\subparagraph*{Conditional Decidability Subject to Schanuel's conjecture.}
Consider the set
	\begin{equation*}
	\mathfrak{S}:= \bigl\{\pi, \pi\iu, \pi S, \pi S\iu, \eu^{\pi}, \eu^{\pi \iu}, \eu^{\pi S}, \eu^{\pi S\iu} \bigr\}
	\end{equation*}
with cardinality \(4m+4\).
We observe that the elements in the subset  \(\{\pi, \pi\iu, \pi S, \pi S\iu\} \subset \mathfrak{S}\) 
are \(\Q\)-linearly independent.
It follows that if Schanuel's conjecture is true, then \(\mathfrak{S}\) possesses a subset of cardinality at least \(2m+2\) whose elements are algebraically independent.
By construction, this algebraically independent subset is necessarily 	
\(	\{\pi, \eu^{\pi}, \eu^{\pi S}, \eu^{\pi S\iu}\}\)
 since the \(2m+2\) elements of \(\{\pi, \pi\iu, \pi S, \pi S\iu\}\) are pairwise algebraically dependent and \(\eu^{\pi\iu}=-1\).

We now rewrite the equality in~\eqref{eq:polytest} in terms of the (obvious) polynomial \(\hat{P}\)
 that absorbs the algebraically dependent parameters of \(S\cup S\iu\) into the coefficients. 
 That is to say, we employ a polynomial \(\hat{P} \in \mathbb{L}(X_1,\ldots, X_{2m+2})\) where \(\mathbb{L}\) is the Galois closure of the number field \(\mathbb{Q}(\theta)(S, S\iu)\) and evaluate \(\hat{P}\) on the algebraically independent subset \(	\{\pi, \eu^{\pi}, \eu^{\pi S}, \eu^{\pi S\iu}\}\) of \(\mathfrak{S}\).
 It follows that the equality in~\eqref{eq:polytest} holds only if
	\begin{equation} \label{eq:polytest2}
		\hat{P} ( \pi, \eu^{\pi}, \eu^{\pi S}, \eu^{\pi S\iu} ) = 0.
	\end{equation}
By \autoref{lem:alg2rat}, the equality in~\eqref{eq:polytest2} holds only if there exists a non-trivial polynomial \(Q\in\Q[X_1,\ldots, X_{2m+2}]\) such that \(Q ( \pi, \eu^{\pi}, \eu^{\pi S}, \eu^{\pi S\iu} )=0\).
Recall that if Schanuel's conjecture is true, then the elements of the set \(\{\pi, \eu^{\pi}, \eu^{\pi S}, \eu^{\pi S\iu}\}\)  are algebraically independent over \(\Q\), from which we deduce that the preceding equality cannot hold.

Subject to the truth of Schanuel's conjecture, we can determine equality tests of the above form.
Thus we have conditional decidability of the  Threshold Problem for the desired class of hypergeometric sequences.
\end{proof}

\begin{remark} \label{rem:MacintyreWilkie}
We note that the equality test 
	\(
		\mathcal{Q} ( \pi, \eu^{\pi}, \eu^{\pi S}, \eu^{\pi S\iu} )=0
	\) can be realised as a proposition in the first-order theory of the reals with exponentiation.
Macintyre and Wilkie \cite{macintyre1996decidability} established decidability of said theory subject to the truth of Schanuel's conjecture.
As noted in previous works, careful inspection of Macintyre and Wilkie's algorithm reveals that correctness is independent of the truth of Schanuel's conjecture.
Indeed, Schanuel's conjecture is only used to prove  termination.
Thus if we apply Macintyre and Wilkie's  algorithm to determine whether the equality  \(\mathcal{Q} ( \pi, \eu^{\pi}, \eu^{\pi S}, \eu^{\pi S\iu} )=0\) holds and find the procedure terminates, then the output is certainly correct.

We note that Macintyre and Wilkie's  algorithm terminates unless the inputs constitute a counterexample to Schanuel's conjecture.
Thus, the process underlying the proof of \autoref{thm:mainclass} presents an interesting prospect  in the sense described by Richardson in \cite{richardson1992elementary} (see also \cite{richardson1997recognize}) ``A failure of the [process] to terminate would be even more interesting than [its] success.''
\end{remark}

\begin{remark} \label{rem:target2}
 Recall \autoref{rem:target} where we extended our class of problem instances to include setups with \(u_0,t\in\mathbb{L}(\pi,\eu^{\pi})\).
 Under the assumption that Schanuel's conjecture is true we can include a broader range of setups.
 Notice, for instance, that the algebraic independence of \(\pi\) and \(\eu\) is currently unknown; however, if Schanuel's conjecture is true, then it follows that \(\pi\) and \(\eu\) are algebraically independent.
 	Thus, subject to the truth of Schanuel's conjecture, we can extend our results in \autoref{prop:maingauss},  \autoref{thm:mainquadratic}, and \autoref{thm:mainclass} to instances where \(u_0, t\in\mathbb{L}(\pi, \eu)\).
 	In fact, we can go further in this direction since the truth of Schanuel's conjecture implies the algebraic independence of the numbers \(\eu\), \(\eu^{\pi}\), \(\eu^{\eu}\), \(\pi\), \(\pi^\pi\), \(\pi^\eu\), \(2^{\pi}\), \(2^\eu\), \(\log \pi\), \(\log 2\), \(\log 3\), \(\log\log 2\), \((\log 2)^{\log 3}\), \(2^{\sqrt{2}}\), and many more (cf.~\cite[Conjecture $S_7$]{ribenboim2000numbers}). 
\end{remark}

\section{Conclusion} \label{sec:conclusion}

\subparagraph*{Summary.}

In this paper we establish (un)conditional decidability results for the Threshold Problem for hypergeometric sequences and, as a side-effect, (un)conditional decidability results for the Membership Problem for hypergeometric sequences.
Previous works have considered the Membership Problem for hypergeometric sequences~\cite{kenison2020positivity,nosan2022membership}; however, the approach in those works cannot handle instances of the Threshold Problem.
The novelty of our approach is the combination of a classical convergence result (\autoref{thm:productgamma}) with results on the algebraic independence of common mathematical constants.

\subparagraph*{Obstacles.}

Let us illustrate an immediate obstacle to the methods herein.
We cannot handle parameters drawn from biquadratic fields because the monic polynomials that split over such fields are not necessarily amenable to our approach.
Recall that biquadratic fields are a particularly well-behaved class of quartic fields (such as \(\Q(\sqrt{5},\sqrt{13})\) and \(\Q(\sqrt{21},\sqrt{33})\)).
For example the minimal polynomial \(x^4 -5x^3 -71x^2 +120x +1044\) of \((5+3\sqrt{5} + \sqrt{13} + 3\sqrt{65})/4\in \Q(\sqrt{5},\sqrt{13})\) does not have \autoref{ass:symmetry}.
Similarly, the minimal polynomial \(x^4 -x^3 -16x^2 +37x -17\) of \((1+\sqrt{21} +\sqrt{33} - \sqrt{77})/4 \in\Q(\sqrt{21},\sqrt{33})\)
does not have \autoref{ass:symmetry}.
Both of these examples are taken from \cite{williams1970integers}.

We also note that the class of sequences we can handle does not permit standard operations on the parameters such as addition.
Consider, for example, that \(\sqrt{2}\) and \(\sqrt[4]{2}\) are both unnested radicals whose minimal polynomials satisfy \autoref{ass:symmetry}; however, the minimal polynomial \(x^4 -4x^2-8x+2\) of \(\sqrt{2} + \sqrt[4]{2}\) does not possess \autoref{ass:symmetry}.

It is not clear how to extend our approach to hypergeometric sequences with larger classes of parameters.
For example, the parameters herein are all algebraic integers.
Even in the restricted setting of hypergeometric sequences with rational parameters (as in the work of Nosan et al.~\cite{nosan2022membership}) it is beyond the state of the art to evaluate equalities between associated gamma products.
Indeed, for \(s\in\{1/6, 1/4, 1/3, 2/3, 3/4\}\) and \(n\in\N\), it is known that \(\Gamma(n + s)\)
is a transcendental number and algebraically independent of \(\pi\) (cf.\ \cite{Waldschmidt2006}).
However, transcendence of the gamma function at other rational points is not known.
It is notable that for rational parameters determining equality between gamma products is decidable subject to the truth of the Rohrlich--Lang conjecture (which itself concerns multiplicative relations for the gamma function~\cite{lang1966introduction,Waldschmidt2006}).

\subparagraph*{Directions for Future Work.}

We give one class of hypergeometric sequences whose parameters link the related works.
Consider the class of hypergeometric sequences whose parameters lie in \(\Q(\iu)\).
For such sequences, decidability of both the Membership and Threshold Problems is open.

The sequences in this class generalise the setting discussed here (and in work by Kenison et al.\ \cite{kenison2023membership}) by removing the condition that the polynomial coefficients of the defining recurrence relation are both monic.
Further, results in this direction would extend the discussion of the Membership Problem for hypergeometric sequences with rational parameters in work by Nosan et al.~\cite{nosan2022membership} as well as those sequences whose polynomial coefficients have irreducible factors in \(\mathcal{C}\) (i.e., each irreducible factor has a root with rational real part) discussed herein.

\bibliography{hypergeombib}

\appendix
\section{Appendixed Proofs}
\label{app:turing}

\lemMPreduction*
\begin{proof}
When \(r(n)\equiv 0\) decidability of both problems is trivial.
So we assume that this is not the case for the remainder of the proof.
Write the shift quotient \(r(n)= c\frac{q(n)}{p(n)}\) where \(p,q\in\Q[x]\) are monic polynomials and \(c\in\Q\).
Without loss of generality, we can assume that \(c>0\); for otherwise, sequence \(\seq[\infty]{u_n}{n=0}\) is given by the interlacing of two hypergeometric sequences with this property.
Let us assume that \(r(n)\) diverges to \(+\infty\).
In this case it is easily seen that, for each \(t\in\Q\), there exists a computable \(N_0\in\N\) such that if \(n\ge N_0\) then \( |u_n| = |u_0 \cdot \prod_{k=0}^n r(k)|> |t| \).
Thus to determine the Threshold Problem in this instance, we need only determine the ultimate sign of  \(\seq[\infty]{u_n}{n=0}\) (which is straightforward).
Moreover, we can compute a bound \(N_1\)  after which the sign of \(\seq[\infty]{u_n}{n=0}\) is constant.
	Thus the Threshold Problem in such instances reduces to an exhaustive search that asks whether \(u_n\ge t\) for each \(n\in\{0, 1, \ldots, \max\{N_0,N_1\}\}\).
	Mutatis mutandis, decidability is similarly established for instances of the Threshold Problem where \(r(n)\) converges to a limit \(\ell\) with \(|\ell|\neq 1\).
	
	The argument for the Membership Problem is similar and given in full in \cite{nosan2022membership}.
\end{proof}

\propreductiontogamma*
\begin{proof}%
From \autoref{lem:MPreduction}, we need only consider cases where the associated shift quotient \(r(k)\) converges to \(\pm 1\) and, by \autoref{thm:productgamma}, we can assume without loss of generality that \(r(k)\) is harmonious.
We treat the case that the sequence of partial products \(\seq{\prod_{k=0}^n r(k)}{n}\) is eventually strictly increasing.
The case where the sequence of partial products is eventually strictly decreasing follows \emph{mutatis mutandis}.

Consider an instance \((\seq{u_n}{n}, t)\) of the Threshold Problem with \(r(k)\) as above and recall our running assumption that \(u_0=1\).
Let \(\tau := \prod_{k=0}^\infty r(k)\).
We assume that the sequence \(\seq{u_n}{n}\), whose terms \(u_n = \prod_{k=0}^n r(k)\) are given by partial products, is eventually strictly increasing.
Then there exists a computable \(N\in\N\) such that \(u_n < \tau\) for each \(n\ge N\).
There are two subcases to consider.
First, if \(\tau \le t\) then it is clear that \(u_n < t\) for each \(n\ge N\) and so we return the answer \texttt{no} to the Threshold Problem.
Second, if \(\tau > t\) then there exists an \(N_1\in\N\) such that \(u_n > t\) for each \(n\ge N_1\).
So decidability of Threshold in this instance reduces to an exhaustive check that asks whether \(u_n\ge t\) for each \(n\in\{0, 1, \ldots, {N-1}\}\).

All that remains is to decide whether \(\tau \le t\).
It is clear that, by computing \(\tau\) to sufficient precision, the problem of determining whether \(\tau<t\) or \(\tau>t\) is recursively enumerable.
Thus we need only test whether the equality \(\tau=t\) holds.
By \autoref{thm:productgamma}, we know that \(\tau = \prod_{j=1}^m {\Gamma(\beta_j)}/{\Gamma(\alpha_j)}\), 
from which we deduce the desired result.

For the sake of brevity, we omit the argument for the reduction from the Membership Problem, which is
near identical to the reasoning displayed above.
\end{proof}

\propreduction*
\begin{proof}%
By \autoref{thm:productgamma} and \autoref{prop:reductiontogamma2}, we need only consider hypergeometric sequences with harmonious shift quotients.
Thus, we continue under this assumption.

Let \((\seq{u_n}{n}, t)\) be an instance of the Membership Problem as above and, in addition, assume that \(\seq{u_n}{n}\) is eventually decreasing.
We note there is a computable bound \(N_0\in\N\) such that \(u_{n+1}\le u_{n}\)  for all \(n\ge N_0\).
Now let \(\tau\) be the limit of the the sequence \(\seq{u_n}{n}\).
Either we have that \(\tau=t\) or \(\tau\neq t\).
We note that the Membership Problem is decidable when \(\tau\neq t\) and so it remains to test cases when \(\tau=t\).
\begin{itemize}
\item Suppose that an oracle for the Threshold Problem returns the answer \texttt{yes} to the problem instance \((\seq[\infty]{u_n}{n=N_0}, \tau)\).
Since \(u_{n+1} < u_n\) for all \(n\ge N_0\), we deduce that \(\tau\) is not a member of the sequence \(\seq[\infty]{u_n}{n=N_0}\).
Thus all that remains is to test whether \(\tau\in\{u_0, u_1, \ldots, u_{N_0-1}\}\).

\item Suppose that an oracle for the Threshold Problem returns the answer \texttt{no} to the problem instance \((\seq[\infty]{u_n}{n=N_0}, \tau)\).
Then there is a computable bound \(N_1\) such that \(u_n < \tau\) for all \(n\ge N_1\).
Thus we can decide the Membership Problem by testing whether \(\tau\in\{u_0,u_1,\ldots, u_{N_1-1}\}\).
\end{itemize}

We note a similar argument to that given above holds for testing the Membership Problem for hypergeometric sequences that are eventually increasing.
Thus we deduce the desired result: that decidability of the Membership Problem for hypergeometric sequences Turing-reduces to that of the Threshold Problem for hypergeometric sequences.
\end{proof}
\section{Threshold for the Kepler--Bouwkamp Constant Approximation}
\label{app:KB}

In this section we demonstrate that we can conditionally determine the Threshold Problem for  instances \((\seq[\infty]{v_n}{n=3}, t)\) where sequence \(\seq[\infty]{v_n}{n=3}\) is the recurrence sequence associated with the approximation of the Kepler--Bouwkamp Constant (\autoref{ex:KB}) by the \([2,2]\)-Pad\'{e} approximant of \(\cos(\wc)\).

Let us begin.
Recall that the sequence \(\seq[\infty]{v_n}{n=3}\) has terms given by \(v_n := \prod_{k=3}^n \frac{12k^2-5\pi^2}{12k^2 + \pi^2}\).
By \cref{prop:reductiontogamma2}, decidability of problem instance \((\seq[\infty]{v_n}{n=3}, t)\) reduces to determining whether 

				\begin{equation} \label{eq:BW}
					\prod_{k=3}^\infty \frac{12k^2- 5\pi^2}{12k^2 + \pi^2} = \frac{\Gamma(3- \tfrac{\iu}{6}\sqrt{3}\pi)\Gamma(3+\tfrac{\iu}{6}\sqrt{3}\pi)}{\Gamma(3-\frac{1}{6}\sqrt{15}\pi)\Gamma(3+\frac{1}{6}\sqrt{15}\pi)}
				\end{equation} 
				is equal to \(t\). (We note that the formulation as a gamma product is given to us by \autoref{thm:productgamma}.)
				
We consider the numerator and denominator in turn.
First, we use the translation and reflection properties of the gamma function to write the numerator of \eqref{eq:BW} as
\begin{align}
	\Gamma \mleft(3- \frac{\sqrt{3}\pi\iu}{6}\mright)\Gamma\mleft(3+\frac{\sqrt{3}\pi\iu}{6}\mright)
	&= \Gamma\mleft(-\frac{\sqrt{3}\pi\iu}{6}\mright)\Gamma\mleft(\frac{\sqrt{3}\pi\iu}{6}\mright) \prod_{k=0}^2 \mleft( k^2 + \frac{\pi^2}{12} \mright) \nonumber \\
	&= \frac{4\sqrt{3}}{\eu^{\frac{\pi^2}{2\sqrt{3}}} - \eu^{-\frac{\pi^2}{2\sqrt{3}}}} \prod_{k=0}^2 \mleft( k^2 + \frac{\pi^2}{12} \mright) \nonumber \\
	&= \frac{4\sqrt{3} \eu^{\frac{\pi^2}{2\sqrt{3}}}}{\eu^{\frac{\pi^2}{\sqrt{3}}} - 1} \prod_{k=0}^2 \mleft( k^2 + \frac{\pi^2}{12} \mright).  \label{eq:numerator}
\end{align}

Second, we likewise write the denominator of \eqref{eq:BW} as
\begin{align}
	{\Gamma\mleft(3-\frac{\sqrt{15}\pi}{6}\mright)\Gamma\mleft(3+\frac{\sqrt{15}\pi}{6}\mright)} 
	&= {\Gamma\mleft(-\frac{\sqrt{15}\pi}{6}\mright)\Gamma\mleft(\frac{\sqrt{15}\pi}{6}\mright)} \prod_{k=0}^2 \mleft( k^2- \frac{5\pi^2}{12}\mright) \nonumber \\
	&=  -\frac{12\iu}{\sqrt{15} \left(\eu^{\frac{\sqrt{15}\pi^2\iu}{6}} - \eu^{-\frac{\sqrt{15}\pi^2\iu}{6}} \right)}
	\prod_{k=0}^2 \mleft( k^2- \frac{5\pi^2}{12}\mright) \nonumber \\
	&= -\frac{12\iu \eu^{\frac{\pi^2\sqrt{15}\iu}{6}} }{\sqrt{15} \left(\eu^{\frac{\pi^2\sqrt{15}\iu}{3}} - 1 \right)}
	\prod_{k=0}^2 \mleft( k^2- \frac{5\pi^2}{12}\mright). \label{eq:denominator}
\end{align}

Let us consider the decision problem at hand.
Taken together, we are tasked to determine whether the ratio of \eqref{eq:numerator} and \eqref{eq:denominator} is equal to \(t\).
For this equality to hold, a simple rearrangement argument leads us to the following:
there is a non-trivial polynomial \(\mathcal{Q}\in\Q[x,y,z]=0\) such that 
\(\mathcal{Q} \bigl( \pi^2, \eu^{\frac{\pi^2}{\sqrt{3}}}, \eu^{\frac{\sqrt{15}\pi^2\iu}{3}} \bigr)=0\) (the fact that such a polynomial with rational coefficients exists is guaranteed by \autoref{lem:alg2rat}).
However, we claim that no such polynomial \(\mathcal{Q}\) exists if Schanuel's conjecture is true. 

Let us prove the above claim.
Consider the set
	\begin{equation*}
	 \mathfrak{S}:= \mleft\{ \tfrac{\pi^2}{\sqrt{3}}, \tfrac{\sqrt{15}\pi^2\iu}{3}, \pi\iu,	\eu^{\frac{\pi^2}{\sqrt{3}}}, \eu^{\frac{\sqrt{15}\pi^2\iu}{3}}, \eu^{\pi\iu} 	\mright\}.
	 \end{equation*}
We note that the elements of  \(\bigl\{\tfrac{\pi^2}{\sqrt{3}}, \tfrac{\sqrt{15}\pi^2\iu}{3}, \pi\iu \bigr\}\) are \(\Q\)-linearly independent.
Thus, Schanuel's conjecture predicts that there is a subset of \(\mathfrak{S}\) of size at least \(3\) whose elements are algebraically independent.
Since \(\tfrac{\pi^2}{\sqrt{3}}\) and \(\tfrac{\sqrt{15}\pi^2\iu}{3}\) are algebraically dependent and \(\eu^{\pi\iu}=-1\), we deduce that, subject to the truth of Schanuel's conjecture, the elements of \(\bigl\{ \pi^2, \eu^{\frac{\pi^2}{\sqrt{3}}}, \eu^{\frac{\sqrt{15}\pi^2\iu}{3}} \bigr\}\) are algebraically independent.
From the preceding work, we deduce that there is no non-trivial polynomial \(\mathcal{Q}\) for which \(\mathcal{Q} \bigl( \pi^2, \eu^{\frac{\pi^2}{\sqrt{3}}}, \eu^{\frac{\sqrt{15}\pi^2\iu}{3}} \bigr)=0\).
It follows that the desired equality cannot hold.
Thus we can conditionally decide the Threshold Problem for instances \((\seq[\infty]{v_n}{n=0}, t)\) with \(t\in\Q\).

\end{document}